\documentclass[a4paper, 10pt]{article}

\usepackage{bbm,xspace}

\usepackage{amsmath}
\usepackage{amssymb}
\usepackage{amsthm}

\usepackage[margin=3.2cm]{geometry}

\theoremstyle{plain}
\newtheorem{theorem}{Theorem}[section]

\theoremstyle{definition}

\newtheorem{remark}{Remark}[section]

\newcommand{\as}{\\[.6em]}

\newcommand{\bela}[1]{\begin{equation}\label{#1}}
\newcommand{\ela}{\end{equation}}
\newcommand{\bear}[1]{\begin{array}{#1}}
\newcommand{\ear}{\end{array}}
\newcommand{\del}{\partial}

\newcommand{\pf}{\operatorname{pf}}
\newcommand{\Gr}{\operatorname{Gr}}
\newcommand{\rank}{\operatorname{rank}}
\newcommand{\one}{\mathbbm{1}}

\newcommand{\Z}{\mathbbm{Z}}
\newcommand{\TED}{TED\xspace}

\numberwithin{equation}{section}

\title{On an integrable multi-dimensionally consistent 
$2n$+$2n$-dimensional heavenly-type equation}

\author{B.G.\ Konopelchenko$^{1}$ and W.K.\ Schief$^{2}$
\bigskip\\  
$^{1}$Department of Mathematics and Physics,\\ University of Salento and 
INFN, Sezione di Lecce,\\ Lecce, 73100, Italy
\bigskip\\
$^{2}$School of Mathematics and Statistics,\\ The University of New South Wales,\\ Sydney, NSW 2052, Australia}

\begin{document}

\maketitle

\begin{abstract}
Based on the commutativity of scalar vector fields, an algebraic scheme is developed which leads to a privileged multi-dimensionally consistent $2n$+$2n$-dimensional integrable partial differential equation with the associated eigenfunction constituting an infinitesimal symmetry. The ``universal'' character of this novel equation of vanishing Pfaffian type is demonstrated by retrieving and generalising to higher dimensions a great variety of well-known integrable equations such as the dispersionless KP and Hirota equations and various avatars of the heavenly equation governing self-dual Einstein spaces.
\end{abstract}

\section{Introduction}

The pioneering work of Zakharov and Shabat \cite{ZakharovShabat1979} on integrable systems which are encoded in the commutativity requirement of a pair of scalar vector fields continues to generate considerable activity in this area of integrable systems theory (see, e.g., \mbox{\cite{BogdanovKonopelchenko2013,ManakovSantini2006}} and references therein). Nevertheless, while, from a constructive point of view, it is transparent why integrable systems of this type exist in {\em any} dimension, it is by no means evident how large the class of multi-dimensional {\em single} equations is which may be obtained in this manner. Prominent examples in three dimensions are given by the dispersionless KP and Hirota equations. The latter equation has the distinct property \cite{Krynski2018} that it is ``compatible with itself'' in the sense that any number of dispersionless Hirota equations, each containing three independent variables, may be solved simultaneously without having to constrain the Cauchy data associated with any individual member of this system of equations. 
This kind of phenomenon known as {\em multi-dimensional consistency} \cite{NijhoffWalker2001,BobenkoSuris2002} has proven to be a powerful key indicator of integrability in the discrete setting. This  raises the question as to the extent to which the two important areas of multi-dimensional consistency and multi-dimensional integrable equations in the spirit of Zakharov and Shabat may be brought together. 

The existence of multi-dimensionally consistent discrete equations in dimensions higher than three remains an open problem. The resolution of this problem is widely regarded as being synonymous with the determination of whether or not higher-dimensional discrete integrable equations exist. Here, it is important to note that not all equations of the compatible system have to be of the same type (cf.\ \cite{AdlerBobenkoSuris2012}). Multi-dimensional consistency in the strict sense of ``compatibility with itself'' of higher-dimensional nonlinear partial differential equations has been observed but this appears to be a very rare phenomenon. In fact, it was only in 2015 that it was indicated \cite{Bogdanov2015} that a four-dimensional integrable equation first recorded in \cite{Schief1999,Schief1996} and termed {\em general heavenly equation} in \cite{DoubrovFerapontov2010} is, in fact, multi-dimensionally consistent. 

In this paper, we show that the general heavenly equation constitutes a travelling wave reduction of the novel integrable 4+4-dimensional equation
\bela{I1}
  \epsilon^{iklm}\Theta_{x^iy^k}\Theta_{x^ly^m} = 0
\ela
which is demonstrated to be likewise multi-dimensionally consistent. Here, $\epsilon^{iklm}$ denotes the totally anti-symmetric Levi-Civita symbol and Einstein's summation convention over repeated indices has been adopted. Remarkably, the same applies to its natural $2n$+$2n$-dimensional generalisation which may be formulated as the vanishing Pfaffian of a $2n$$\times$$2n$ skew-symmetric matrix $\omega$, namely
\bela{I2}
 \pf(\omega)=0,\quad \omega_{ik} = \Theta_{x^iy^k} - \Theta_{x^ky^i}.
\ela
In the following, we refer to the partial differential equation (\ref{I2}) as the {\em \TED equation} for reasons to be explained in \S3(b).

The algebraic scheme developed here not only leads to the \TED equation but also to the identification of the associated eigenfunction as an infinitesimal symmetry of the \TED equation. In fact, a close connection between this infinitesimal symmetry and the multi-dimensional consistency of the \TED equation is revealed. Moreover, the ``universality'' of the \TED equation is demonstrated by not only capturing, for instance, the important dispersionless KP and Hirota equations but also a variety of well-known heavenly equations governing self-dual Einstein spaces (see, e.g., \cite{DoubrovFerapontov2010,PlebanskiPrzanowski1996}). This, in turn, may be exploited to derive higher-dimensional integrable analogues of those equations. We also briefly mention connections with Grassmannians, K\"ahler geometry and the multi-dimensional Manakov-Santini integrable system \cite{ManakovSantini2006} which open up interesting directions for future research. 

\section{The general heavenly equation}

The partial differential equation
\bela{Z1}
  (\lambda_3 - \lambda_4)\Theta_{x^1x^2}\Theta_{x^3x^4} = \lambda_3\Theta_{x^2x^3}\Theta_{x^1x^4} - \lambda_4\Theta_{x^1x^3}\Theta_{x^2x^4}
\ela
has originally been introduced \cite{Schief1999} in connection with a permutability theorem for the integrable Tzitz\'eica equation of affine differential geometry \cite{BobenkoSchief1999} and has been termed {\em general heavenly equation} in \cite{DoubrovFerapontov2010}.
It guarantees the compatibility of the Lax pair
\bela{Z2}
 \begin{split}
  \psi_{x^3} & = \frac{1}{(\lambda-\lambda_3)\Theta_{x^1x^2}}(\lambda\Theta_{x^1x^3}\psi_{x^2} - \lambda_3\Theta_{x^2x^3}\psi_{x^1})\\
  \psi_{x^4} & = \frac{1}{(\lambda-\lambda_4)\Theta_{x^1x^2}}(\lambda\Theta_{x^1x^4}\psi_{x^2} - \lambda_4\Theta_{x^2x^4}\psi_{x^1}),
 \end{split}
\ela
wherein $\lambda$ is an arbitrary (spectral) parameter. Moreover, an appropriately chosen pair of eigenfunctions $\psi$ has been shown to encapsulate Plebanski's first heavenly equation \cite{Schief1996}. 
An equivalent symmetric formulation of the general heavenly equation is given by
\bela{Z5}
  \alpha\Theta_{x^1x^2}\Theta_{x^3x^4} +\beta\Theta_{x^2x^3}\Theta_{x^1x^4} + \gamma\Theta_{x^3x^1}\Theta_{x^2x^4} = 0,\quad \alpha+\beta+\gamma=0
\ela
or, explicitly,
\bela{Z6}
 \begin{split}
  & \phantom{+}\,\,\, (\lambda_1-\lambda_2)(\lambda_3-\lambda_4)\Theta_{x^1x^2}\Theta_{x^3x^4}\\
  & +  (\lambda_2-\lambda_3)(\lambda_1-\lambda_4)\Theta_{x^2x^3}\Theta_{x^1x^4}\\  
  & +  (\lambda_3-\lambda_1)(\lambda_2-\lambda_4)\Theta_{x^3x^1}\Theta_{x^2x^4} = 0.
\end{split}
\ela
The avatar (\ref{Z1}) is obtained by letting $\lambda_1\rightarrow\infty$ and $\lambda_2=0$, corresponding to a fractional linear transformation of the parameters $\lambda_i$. Based on a general scheme developed in \cite{BogdanovKonopelchenko2013,BogdanovKonopelchenko2014}, it has been observed only recently \cite{Bogdanov2015} that, remarkably, the general heavenly equation in the form (\ref{Z6}) may be consistently extended to an arbitrarily large system of equations, each of which constitutes a general heavenly equation which involves a subset of four independent variables and associated parameters, namely
\bela{Z7}
 \begin{split}
  & \phantom{+}\,\,\, (\lambda_i-\lambda_k)(\lambda_l-\lambda_m)\Theta_{x^ix^k}\Theta_{x^lx^m}\\
  & +  (\lambda_k-\lambda_l)(\lambda_i-\lambda_m)\Theta_{x^kx^l}\Theta_{x^ix^m}\\  
  & +  (\lambda_l-\lambda_i)(\lambda_k-\lambda_m)\Theta_{x^lx^i}\Theta_{x^kx^m} = 0,
\end{split}
\ela
where $i,k,l,m$ are distinct but otherwise arbitrary integers. It will be shown in \S4 that this type of {\em multi-dimensional consistency} constitutes a special case of that associated with the 4+4-dimensional extension (\ref{I1}) of the general heavenly equation and, in fact, its canonical analogue (\ref{I2}) in arbitrary $2n$+$2n$ dimensions.

The symmetry of the general heavenly equation indicates that its Lax pair encodes four linear equations $X_i\psi=0$, $i=1,2,3,4$ of which only two are independent. Indeed, it is readily verified that the Lax pair (\ref{Z2}) represents the third and fourth equation of the linear system
\bela{Z8}
  \left(\bear{cccc} 0 & \omega_{34} & \omega_{42} & \omega_{23}\\ \omega_{43} & 0 & \omega_{14} & \omega_{31}\\
                    \omega_{24} & \omega_{41} & 0 & \omega_{12}\\ \omega_{32} & \omega_{13} & \omega_{21} & 0\ear\right)
  \left(\bear{c}D_1\psi\\ D_2\psi\\ D_3\psi\\ D_4\psi\ear\right)  = 0,
\ela
where the operators $D_i$ are given by
\bela{Z9}
 D_1 = \del_{x^1},\quad D_2 = -\lambda\del_{x^2},\quad D_3 = (\lambda_3-\lambda)\del_{x^3},\quad D_4 = (\lambda_4-\lambda)\del_{x^4}
\ela
and the coefficients $\omega_{ik} = -\omega_{ki}$ forming a skew-symmetric matrix $\omega$ read
\bela{Z10}
 \begin{aligned}
  \omega_{12} &= -\Theta_{x^1x^2},\quad &\omega_{23} &= \lambda_3\Theta_{x^2x^3},\quad & \omega_{31} &=\Theta_{x^3x^1}\\
  \omega_{14} &= -\Theta_{x^1x^4},\quad &\omega_{24} &= \lambda_4\Theta_{x^2x^4},\quad & \omega_{34} &=(\lambda_4-\lambda_3)\Theta_{x^3x^4}.
 \end{aligned}
\ela
Moreover, the general heavenly equation may be formulated as the vanishing Pfaffian
\bela{Z11}
  \pf(\omega) = \omega_{12}\omega_{34} + \omega_{23}\omega_{14} + \omega_{31}\omega_{24} = 0.
\ela
Since $\det\omega = [\pf(\omega)]^2$ (see, e.g., \cite{Hirota2003}), the latter is equivalent to the condition that $\rank\omega=2$. Hence, the first two equations of the linear system (\ref{Z8}), which may be formulated as
\bela{Z12}
  A^{lm}D_m\psi = 0,\quad A^{lm} = \epsilon^{iklm}\omega_{ik},
\ela
are redundant as asserted. It is recalled that the rank of a skew-symmetric matrix is even. Even-dimensional skew-symmetric matrices $A$ which are privileged in that a single condition on the entries $A^{lm}$ is equivalent to $\rank A=2$ turn out to play the central role in  this paper.

\section{The 4+4-dimensional \TED equation}

In this section, we identify the first principles which govern the rich algebraic structure summarised in the preceding and give rise to a general scheme which encapsulates the general heavenly equation as a particular case. 

\subsection*{(a) The algebraic scheme}

We begin with a linear system of the form 
\bela{Z13}
  \left(\bear{cccc} 0 & \omega_{34} & \omega_{42} & \omega_{23}\\ \omega_{43} & 0 & \omega_{14} & \omega_{31}\\
                    \omega_{24} & \omega_{41} & 0 & \omega_{12}\\ \omega_{32} & \omega_{13} & \omega_{21} & 0\ear\right)
  \left(\bear{c}D_1\psi\\ D_2\psi\\ D_3\psi\\ D_4\psi\ear\right)  = 0,
\ela
wherein the $D_i$ denote commuting linear operators which obey the standard product rule $D_i(fg) = (D_if)g + fD_ig$ and the matrix $\omega$ is skew-symmetric. In order for the linear system to have rank 2, we impose the condition $\pf(\omega)=0$, that is,
\bela{Z14}
  \omega_{12}\omega_{34} + \omega_{23}\omega_{14} + \omega_{31}\omega_{24} = 0.
\ela
Accordingly, generically, the system (\ref{Z13}) is equivalent to the pair of equations
\bela{Z15}
 \begin{split}
  X_3\psi & = (\omega_{32}D_1 + \omega_{13}D_2 + \omega_{21}D_3)\psi = 0\\
  X_4\psi & = (\omega_{24}D_1 + \omega_{41}D_2 + \omega_{12}D_4)\psi = 0.
 \end{split}
\ela
The associated compatibility condition $[X_3,X_4]\psi = 0$ modulo (\ref{Z15}) leads to a condition on $\psi$ which is of first order. It is most conveniently obtained by considering the symmetric system (\ref{Z13}) written as in the previous section as
\bela{Z16}
  A^{lm}D_m\psi = 0,\quad A^{lm} = \epsilon^{iklm}\omega_{ik}.
\ela
Indeed, we immediately deduce that
\bela{Z17}
 0 = D_l (A^{lm}D_m\psi) =  D_l(\epsilon^{iklm}\omega_{ik}D_m\psi)  = \epsilon^{iklm}(D_l\omega_{ik})D_m\psi
\ela
is, modulo the linear system itself, equivalent to the compatibility condition.

It is evident that the compatibility condition (\ref{Z17}) is {\em identically} satisfied if we impose the ``vanishing divergence'' conditions
\bela{Z18}
 D_{[l}\omega_{ik]} = 0,
\ela
where the square brackets denote total anti-symmetrisation. These may be formally resolved by setting $\omega_{ik} = D_k\Theta_i - D_i\Theta_k$ for four arbitrary functions $\Theta_i$ and the vanishing Pfaffian condition (\ref{Z14}) becomes a first-order equation for these functions. It is noted that the vector fields $A^{lm}D_m$ are indeed divergence free since $D_mA^{lm}=0$. The latter is also equivalent to the condition that the vector fields be formally anti-self adjoint.

\subsection*{(b) The 4+4-dimensional \TED equation}

In order to obtain an equation for a single function, it is natural to set
\bela{Z19}
 D_i = \del_{y^i} - \lambda\del_{x^i},
\ela
where $\lambda$ is an arbitrary parameter, so that the conditions (\ref{Z18}) split into the two sets
\bela{Z20}
 \del_{x^{[l}}\omega_{ik]} = 0,\quad \del_{y^{[l}}\omega_{ik]} = 0,
\ela
provided that it is assumed that the functions $\omega_{ik}$ are independent of the parameter $\lambda$. These two sets of conditions may be regarded as the integrability conditions for the existence of a function $\Theta$ which parametrises the functions $\omega_{ik}$ according to
\bela{Z21}
  \omega_{ik} = \Theta_{x^iy^k} - \Theta_{x^ky^i}
\ela
so that (\ref{Z14}) assumes the form
\bela{Z22}
 \begin{split}
  & \phantom{+}\,\,\, (\Theta_{x^1y^2} - \Theta_{x^2y^1})(\Theta_{x^3y^4} - \Theta_{x^4y^3})\\  
  & + (\Theta_{x^2y^3} - \Theta_{x^3y^2})(\Theta_{x^1y^4} - \Theta_{x^4y^1})\\
  & + (\Theta_{x^3y^1} - \Theta_{x^1y^3})(\Theta_{x^2y^4} - \Theta_{x^4y^2}) = 0.
 \end{split}
\ela
The latter constitutes a 4+4-dimensional generalisation of the general heavenly equation since the symmetry (travelling wave) reduction
\bela{Z23}
  \Theta_{y^i} = \lambda_i\Theta_{x^i},
\ela
which corresponds to the choice
\bela{Z24}
  D_i = (\lambda_i-\lambda)\del_{x^i},\quad \omega_{ik} = (\lambda_k-\lambda_i)\Theta_{x^ix^k},
\ela
reduces it to the symmetric general heavenly equation (\ref{Z6}). Moreover, the associated symmetric version of the linear system (\ref{Z8}), (\ref{Z9}), (\ref{Z10}) consisting of (\ref{Z13}) and (\ref{Z24}) is retrieved. In summary, we have established the following theorem.

\begin{theorem}
The Lax pair
\bela{Z25}
 \begin{aligned}
  & \phantom{+}\,\,\, (\Theta_{x^1y^2} - \Theta_{x^2y^1})(\psi_{y^3} - \lambda\psi_{x^3})
  & \quad
  & \phantom{+}\,\,\, (\Theta_{x^1y^2} - \Theta_{x^2y^1})(\psi_{y^4} - \lambda\psi_{x^4})\\  
  & + (\Theta_{x^2y^3} - \Theta_{x^3y^2})(\psi_{y^1} - \lambda\psi_{x^1})
  & \quad
  & + (\Theta_{x^2y^4} - \Theta_{x^4y^2})(\psi_{y^1} - \lambda\psi_{x^1})\\
  & + (\Theta_{x^3y^1} - \Theta_{x^1y^3})(\psi_{y^2} - \lambda\psi_{x^2}) = 0,
  & \quad
  & + (\Theta_{x^4y^1} - \Theta_{x^1y^4})(\psi_{y^2} - \lambda\psi_{x^2}) = 0
 \end{aligned}
\ela
is compatible modulo the 4+4-dimensional \TED equation
\bela{Z26}
 \begin{split}
  & \phantom{+}\,\,\, (\Theta_{x^1y^2} - \Theta_{x^2y^1})(\Theta_{x^3y^4} - \Theta_{x^4y^3})\\  
  & + (\Theta_{x^2y^3} - \Theta_{x^3y^2})(\Theta_{x^1y^4} - \Theta_{x^4y^1})\\
  & + (\Theta_{x^3y^1} - \Theta_{x^1y^3})(\Theta_{x^2y^4} - \Theta_{x^4y^2}) = 0.
 \end{split}
\ela
\end{theorem}

It is emphasised that the splitting of the vanishing divergence conditions (\ref{Z18}) into the two sets (\ref{Z20}) may also be achieved by making the choice
\bela{Z26a}
  D_l = \del_{z^l},\quad z^l = x^l +iy^l
\ela
and demanding that the matrix $\omega$ be real so that the 4+4-dimensional \TED equation may be formulated as
\bela{Z26b}
  \epsilon^{iklm}\Theta_{z^i\bar{z}^k}\Theta_{z^l\bar{z}^m} = 0.
\ela
It is noted that the matrix ${(\Theta_{z^i\bar{z}^k})}_{i,k=1,\ldots, 4}$ is known as the {\em complex Hessian} \cite{GunningRossi1965} of $\Theta$ so that one may interpret the 4+4-dimensional \TED equation as the vanishing Pfaffian (of the skew-symmetric part) of the complex Hessian of $\Theta$.

In terms of differential forms, the differential equation (\ref{Z26b}) may be formulated as the pair of equations 
\bela{Z83a}
  \del\Omega = 0,\quad \wedge^2\Omega=0, 
\ela
where the coefficients of the differential 2-form
\bela{Z83b}
  \Omega = \omega_{ik}dz^i\wedge dz^k
\ela
are real and $\del$ is the standard Dolbeault operator \cite{BlumenhagenLuestTheisen2013}, that is,
\bela{Z83c}
 \del\Omega = (\del_{z^l}\omega_{ik})dz^l\wedge dz^i\wedge dz^k.
\ela
Indeed, the condition that $\Omega$ be closed with respect to the Dolbeault operator together with the reality of $\omega$ is equivalent to the vanishing divergence conditions (\ref{Z20}). Here, we assume that the technical assumptions for the validity of the analogue of the Poincar\'e lemma for exterior derivatives hold so that the  [t]wo-form $\Omega$ is [e]xact with respect to the [D]olbeault operator (TED) and, hence, $\Omega \sim \del(\Theta_{\bar{z}^i}dz^i)$ by virtue of the reality of $\omega$. This is the origin of the acronym TED which encodes the main ingredients of the differential form avatar of (\ref{Z26}).  The (reverse) acronym is also indicative of the aforementioned fact that the Pfaffian is the (signed) square root of a [det]erminant so that the \TED equation is equivalent to $\det\omega = [\pf(\omega)]^2=0$.

\subsection*{(c) An infinitesimal symmetry}

In order to establish the multi-dimensional consistency of the 4+4-dimensional \TED equation, we first prove a particular property of the \TED equation, namely that its associated eigenfunction constitutes an infinitesimal symmetry. 

\begin{theorem}\label{symmetry}
Any solution $\psi$ of the Lax pair (\ref{Z25}) or its equivalent symmetric form (\ref{Z13}), (\ref{Z19}), (\ref{Z21}) represents an infinitesimal symmetry of the 4+4-dimensional \TED equation (\ref{Z26}), that is, the latter is preserved by the flow $\Theta_s=\psi$. 
\end{theorem}

\begin{proof}
If we formulate the 4+4-dimensional \TED equation as
\bela{Z27}
  \epsilon^{iklm}\omega_{ik}\omega_{lm} = 0
\ela
with $\omega_{ik}$ defined by (\ref{Z21}) then the condition for a quantity $\varphi$ with associated flow
\bela{Z28}
 \Theta_s = \varphi
\ela
to be an infinitesimal symmetry is given by
\bela{Z29}
 \del_s(\epsilon^{iklm}\omega_{ik}\omega_{lm}) = 0,
\ela
By virtue of the symmetry of the Levi-Civita symbol, this condition may be expressed as
\bela{Z30}
  \epsilon^{iklm}\omega_{ik}\varphi_{x^ly^m} = 0
\ela
since
\bela{Z31}
  \del_s\omega_{ik} = \varphi_{x^iy^k} - \varphi_{x^ky^i}.
\ela
On the other hand, the linear system (\ref{Z16}) with $D_i$ given by (\ref{Z19}) implies that
\bela{Z32}
 \begin{split}
  0 & = \del_{x^l}(\epsilon^{iklm}\omega_{ik}D_m\psi)\\
     & = \epsilon^{iklm}[(\del_{x^l}\omega_{ik})D_m\psi + \omega_{ik}\psi_{x^ly^m}]\\
     & = \epsilon^{iklm}\omega_{ik}\psi_{x^ly^m}
 \end{split}
\ela
by virtue of the vanishing divergence conditions (\ref{Z20})$_1$. Accordingly, $\varphi=\psi$ satisfies the condition (\ref{Z30}) for an infinitesimal symmetry.
\end{proof}

\begin{remark}
The above theorem implies that if we eliminate $\psi$ in the Lax pair (\ref{Z25}) in favour of $\Theta_s$ then the 4+4-dimensional \TED equation and the pair
\bela{Z33}
 \begin{aligned}
  & \phantom{+}\,\,\, (\Theta_{x^1y^2} - \Theta_{x^2y^1})(\Theta_{sy^3} - \lambda\Theta_{sx^3})
  & \quad  
  & \phantom{+}\,\,\, (\Theta_{x^1y^2} - \Theta_{x^2y^1})(\Theta_{sy^4} - \lambda\Theta_{sx^4})\\
  & + (\Theta_{x^2y^3} - \Theta_{x^3y^2})(\Theta_{sy^1} - \lambda\Theta_{sx^1})
  & \quad
  & + (\Theta_{x^2y^4} - \Theta_{x^4y^2})(\Theta_{sy^1} - \lambda\Theta_{sx^1})\\
  & + (\Theta_{x^3y^1} - \Theta_{x^1y^3})(\Theta_{sy^2} - \lambda\Theta_{sx^2}) = 0,
  & \quad
  & + (\Theta_{x^4y^1} - \Theta_{x^1y^4})(\Theta_{sy^2} - \lambda\Theta_{sx^2}) = 0 
 \end{aligned}
\ela
are compatible. Each of the above equations evidently constitutes a 4+3-dimensional travelling wave reduction of a 4+4-dimensional \TED equation. Hence, if we take into account the symmetric form (\ref{Z13}) of the Lax pair then the existence of the infinitesimal symmetry $\psi$ is equivalent to the 4+4-dimensional \TED equation being compatible with four 4+3-dimensional \TED equations. 
\end{remark}

\begin{remark}
If we adopt the notation $s=x^5$ and $\lambda=\lambda_5$ then elimination of $\psi$ in the linear system (\ref{Z13}) and (\ref{Z24}) associated with the symmetric general heavenly equation (\ref{Z6}) produces the additional four compatible general heavenly equations (\ref{Z7}). Hence, the multi-dimensional consistency of the general heavenly equation observed by Bogdanov \cite{Bogdanov2015} is generated by its infinitesimal symmetry $\psi$. This is reminiscent of the standard B\"acklund transformation for the discrete integrable master Hirota equation (see, e.g.,  \cite{AdlerBobenkoSuris2012}) with the infinitesimal symmetry replaced by a discrete symmetry.
\end{remark}

\subsection*{(d) Symmetry constraints. A 3+3-dimensional dispersionless Hirota equation}

Reductions of the compatible system of 4+3-dimensional \TED equations generated by the infinitesimal symmetry $\psi$ may be obtained by matching the latter with any other infinitesimal symmetry. For instance, since the 4+4-dimensional \TED equation may also be formulated as
\bela{Z33a}
  \epsilon^{iklm}\Theta_{x^iy^k}\Theta_{x^ly^m} = 0,
\ela
comparison with the symmetry condition (\ref{Z30}) shows that $\varphi=\Theta$ constitutes another infinitesimal symmetry which, in fact, corresponds to the scaling invariance $\Theta\rightarrow\mu\Theta$ of the 4+4-dimensional \TED equation. Accordingly, if we apply the symmetry constraint $\Theta=\psi$ in the symmetric form (\ref{Z13}) of the Lax pair (\ref{Z25}) or, equivalently, set $\Theta_s = \Theta$ in the four 4+3-dimensional \TED equations of type (\ref{Z33}) then we obtain
\bela{Z33b}
 \begin{split}
  & \phantom{+}\,\,\, (\Theta_{x^iy^k} - \Theta_{x^ky^i})(\Theta_{y^l} - \lambda\Theta_{x^l})\\  
  & + (\Theta_{x^ky^l} - \Theta_{x^ly^k})(\Theta_{y^i} - \lambda\Theta_{x^i})\\
  & + (\Theta_{x^ly^i} - \Theta_{x^iy^l})(\Theta_{y^k} - \lambda\Theta_{x^k}) = 0,
 \end{split}
\ela
where the indices $i,k,l$ are distinct. Even though the parameter $\lambda$  encapsulates the symmetry between the independent variables $x^i$ and $y^i$ modulo $\lambda\rightarrow\lambda^{-1}$, one may, without loss of generality, set $\lambda=0$.

By construction, any of the above equations is an algebraic consequence of the other three compatible equations and so is the original 4+4-dimensional \TED equation (\ref{Z33a}). This feature is the exact analogue of the multi-dimensional consistency of the dispersionless Hirota equation discussed in \S5(d). In fact, each of the above equations constitutes a 3+3-dimensional generalisation of the dispersionless Hirota equation (see, e.g., \cite{Krynski2018}) as may be seen by applying the travelling wave reduction $\Theta_{y^i}=\lambda_i\Theta_{x^i}$. It is recalled that the latter has been shown in the preceding to reduce the 4+4-dimensional \TED equation to the general heavenly equation.

\section{Multi-dimensional consistency}

It turns out that the remarkable compatibility property generated by the infinitesimal symmetry of the 4+4-dimensional \TED equation is still present if one replaces the 4+3-dimensional equations by their underlying 4+4-dimensional counterparts. This kind of completely symmetric compatibility phenomenon is well known \cite{NijhoffWalker2001,BobenkoSuris2002} for some privileged discrete integrable equations such as the discrete Hirota equation (cf.\ \S 8(b)), in which context it has come to be known as {\em multi-dimensional consistency} \cite{AdlerBobenkoSuris2012}.

\begin{theorem}
The 4+4-dimensional \TED equation
\bela{Z34}
   H^5 = \omega_{12}\omega_{34} + \omega_{23}\omega_{14} + \omega_{31}\omega_{24} = 0
\ela
is multi-dimensionally consistent, that is, the 5+5-dimensional system of 4+4-dimensional \TED equations
\bela{Z35}
 \begin{split}
   H^1 & = \omega_{23}\omega_{45} + \omega_{34}\omega_{25} + \omega_{42}\omega_{35} = 0,\quad
  H^2 = \omega_{34}\omega_{51} + \omega_{45}\omega_{31} + \omega_{53}\omega_{41} = 0\\
   H^3 & = \omega_{45}\omega_{12} + \omega_{51}\omega_{42} + \omega_{14}\omega_{52} = 0,\quad 
  H^4 = \omega_{51}\omega_{23} + \omega_{12}\omega_{53} + \omega_{25}\omega_{13} = 0\\
 \end{split}
\ela
and (\ref{Z34}), where
\bela{Z36}
  \omega_{ik} = \Theta_{x^iy^k} - \Theta_{x^ky^i}
\ela
for  $i\neq k\in\{1,2,3,4,5\}$,
is compatible in that it is in involution in the sense of Riquier-Janet (Cartan-K\"ahler) theory \cite{Seiler2010}.
\end{theorem}

\begin{proof}
We first note that only three of the five equations (\ref{Z34}), (\ref{Z35}) are algebraically independent. This is best seen by regarding the equations $H^i$, $i=1,\ldots,4$ as linear equations for the quantities $\omega_{i5}$. In fact, these four equations may be formulated as
\bela{Z37}
  \left(\bear{cccc} 0 & \omega_{34} & \omega_{42} & \omega_{23}\\ \omega_{43} & 0 & \omega_{14} & \omega_{31}\\
                    \omega_{24} & \omega_{41} & 0 & \omega_{12}\\ \omega_{32} & \omega_{13} & \omega_{21} & 0\ear\right)
  \left(\bear{c}\omega_{15}\\ \omega_{25}\\ \omega_{35}\\ \omega_{45}\ear\right)  = 0,
\ela
which may be obtained by formally replacing in the linear system (\ref{Z13}) the quantities $D_i\psi$ by $\omega_{i5}$. Since the above linear system has rank 2, it is only necessary to demonstrate the compatibility of, for instance, the triple of equations
\bela{Z38}
 H^3=0,\quad H^4=0,\quad H^5=0.
\ela
Now, the equations $H^3=0$ and $H^4=0$ may be regarded as defining the ``evolution'' of $\Theta$ in the $y^5$-direction, that is, we may formulate them as the pair
\bela{Z39}
 \begin{split}
  \Theta_{x^3y^5} & = \Theta_{x^5y^3} + \frac{\omega_{32}}{\omega_{12}}(\Theta_{x^1y^5} - \Theta_{x^5y^1}) +  \frac{\omega_{13}}{\omega_{12}}(\Theta_{x^2y^5} - \Theta_{x^5y^2})\\
  \Theta_{x^4y^5} & = \Theta_{x^5y^4} + \frac{\omega_{42}}{\omega_{12}}(\Theta_{x^1y^5} - \Theta_{x^5y^1}) +  \frac{\omega_{14}}{\omega_{12}}(\Theta_{x^2y^5} - \Theta_{x^5y^2}).
 \end{split}
\ela
It is observed that the latter constitutes an inhomogeneous analogue of the Lax pair (\ref{Z25}) with the quantity $\Theta_{y^5}$ playing the role of the eigenfunction. Accordingly, it is required to show that this pair is compatible and that the remaining equation $H^5=0$ is in involution with respect to the variable $y^5$. 

As in the case of the compatibility of the Lax pair for the 4+4-dimensional \TED  equation, it is convenient to formulate the compatibility condition ${(\Theta_{x^3y^5})}_{x^4}={(\Theta_{x^4y^5})}_{x^3}$ for the pair (\ref{Z39}) in terms of the symmetric system (\ref{Z37}), that is,
\bela{Z40}
  \epsilon^{iklm}\omega_{ik}(\Theta_{x^my^5} - \Theta_{x^5y^m})=0.
\ela
Thus, it is required to show that the compatibility condition
\bela{Z41}
 \begin{split}
 0 &= \del_{x^l}[\epsilon^{iklm}\omega_{ik}(\Theta_{x^my^5} - \Theta_{x^5y^m})]\\ &= \epsilon^{iklm}[(\del_{x^l}\omega_{ik})(\Theta_{x^my^5} - \Theta_{x^5y^m}) - \omega_{ik}\Theta_{x^5x^ly^m}]
\end{split}
\ela
is satisfied modulo $H^5=0$ and its differential consequences not involving $y^5$. The vanishing divergence conditions (\ref{Z20})$_1$ and the symmetry of the Levi-Civita symbol lead to the simplification
\bela{Z42}
 0 = -\frac{1}{2}\epsilon^{iklm}\omega_{ik}\del_{x^5}\omega_{lm} = -\frac{1}{4}\del_{x^5}(\epsilon^{iklm}\omega_{ik}\omega_{lm})
\ela
which holds since this is essentially $\del_{x^5}H^5 = 0$.

In order to show that $H^5$ is preserved by the $y^5$-evolution, we evaluate $\del_{y^5}H^5$ modulo the
vanishing divergence conditions (\ref{Z20})$_2$ to obtain
\bela{Z43}
 \begin{split}
 \del_{y^5}(\epsilon^{iklm}\omega_{ik}\omega_{lm}) &= 2\epsilon^{iklm}\omega_{ik}\del_{y^5}\omega_{lm}\\
 &=4\epsilon^{iklm}\omega_{ik}\del_{y^m}\omega_{l5}\\
 &=4\del_{y^m}(\epsilon^{iklm}\omega_{ik}\omega_{l5}).
 \end{split}
\ela
This completes the proof of the theorem since the last expression in the brackets coincides with the left-hand side of the linear system (\ref{Z37}) representing $H^3=H^4=0$.
\end{proof}

\begin{remark}
The origin of the two relations on which the above proof is based, namely, on the one hand, the equality of the first expression in (\ref{Z41}) and the second expression in (\ref{Z42}) and, on the other hand, the equality of the first and last expressions in (\ref{Z43}), lies in a simple identity which may also be used to prove the multi-dimensional consistency of the higher-dimensional extension of the 4+4-dimensional \TED equation discussed in \S 6. Specifically, if we scale the quantities $H^i$ by appropriate constants so that
\bela{Z44}
  H^p = \epsilon^{iklmp}\omega_{ik}\omega_{lm},\quad p=1,\ldots,5
\ela
then the vanishing divergence conditions (\ref{Z18}) immediately imply that
\bela{Z45}
  D_pH^p=0
\ela
so that 
\bela{Z46}
 \del_{x^p}H^p = 0,\quad \del_{y^p}H^p=0.
\ela
It is now readily verified that the latter coincide with the two above-mentioned key relations.
\end{remark}  

\section{Reductions}

There exists an extensive literature on the equations governing self-dual Einstein spaces or, equivalently, the self-dual Yang-Mills equations with four translational symmetries and the gauge group of volume preserving diffeomorphisms. In fact, various forms of their representation in terms of a single equation have been derived by various authors (see, e.g., \cite{PlebanskiPrzanowski1996,DoubrovFerapontov2010}, and references therein). The first two such "heavenly" equations have been derived by Plebanski in \cite{Plebanski1975}. There exists a 6-dimen\-sional version of Plebanski's second heavenly equation \cite{Takasaki1989,PlebanskiPrzanowski1996} which enjoys the important property that it admits a variety of reductions which have appeared in the literature in various contexts. These include Plebanski's first and second heavenly equations, the evolutionary form of the second heavenly equation and the modified heavenly equation, and the Husain-Park and Grant equations. However, the general heavenly equation cannot be derived from the \mbox{6-dimen}\-sional second heavenly equation. Here, we show that the 6-dimensional second heavenly equation is, in fact, a 2+4-dimensional reduction of the 4+4-dimensional \TED equation so that the latter unifies the heavenly-type equations listed in the preceding. We also embark on a demonstration of the ``universal'' character of the 4+4-dimensional \TED equation by discussing its reductions of a different type to the Boyer-Finley and dispersionless Kadomtsev-Petviashvili (KP) equations (see, e.g., \cite{BogdanovKonopelchenko2013}), the  dispersionless Hirota equation and the Mikhalev equation \cite{Mikhalev1992}. By construction, any of the above-mentioned reductions is accompanied by the corresponding reduction of the Lax pair. 

\subsection*{(a) The 6-dimensional second heavenly equation}

It is evident that the 4+4-dimensional \TED equation and its associated Lax pair remain autonomous if a term quadratic in the independent variables is separated from the function $\Theta$. This may be used effectively to generate dimensional reductions of the 4+4-dimensional \TED equation (\ref{Z26}). For instance, if we set
\bela{Z47}
  \Theta = x^1y^2 + \Xi(x^3,x^4,y^1,y^2,y^3,y^4)
\ela
then the 2+4-dimensional equation
\bela{Z48}
\Xi_{x^3y^4} - \Xi_{x^4y^3} = \Xi_{x^3y^1} \Xi_{x^4y^2} - \Xi_{x^3y^2}\Xi_{x^4y^1}
\ela
is obtained. This is the standard 6-dimensional second heavenly equation. Its standard Lax pair is retrieved by setting $\psi_{x^1}=\psi_{x^2} = 0$ in (\ref{Z25}), leading to
\bela{Z49}
 \begin{split}
   \psi_{y^3} & = \lambda\psi_{x^3} + \Xi_{x^3y^2} \psi_{y^1} - \Xi_{x^3y^1}\psi_{y^2}\\
   \psi_{y^4} & = \lambda\psi_{x^4} + \Xi_{x^4y^2} \psi_{y^1} - \Xi_{x^4y^1}\psi_{y^2}.
 \end{split}
\ela

\subsection*{(b) The Boyer-Finley equation}

The 6-dimensional second heavenly equation admits the obvious 0+4-dimensional reduction
\bela{Z50}
  \Xi = a(y^1,y^2,y^3,y^4)x^3 + b(y^1,y^2,y^3,y^4)x^4
\ela
so that
\bela{Z51}
 a_{y^4}  - b_{y^3} = a_{y^1}b_{y^2} - a_{y^2}b_{y^1}
\ela
and the Lax pair (\ref{Z49}) simplifies to
\bela{Z52}
 \begin{split}
  \psi_{y^3} & = a_{y^2} \psi_{y^1} - a_{y^1}\psi_{y^2}\\
  \psi_{y^4} & = b_{y^2} \psi_{y^1} -  b_{y^1}\psi_{y^2}
 \end{split}
\ela
since the parameter $\lambda$ is redundant and may be set to zero. If we prescribe the dependence of $a$ and $b$ on some variable, say, $y^2$ then the underdetermined equation (\ref{Z51}) gives rise to a multi-component system and $y^2$ may be regarded as a proper ``spectral'' parameter. For instance, if we set
\bela{Z53}
  a = \tilde{\lambda} + p(y^1,y^3,y^4),\quad b = q(y^1,y^3,y^4)/\tilde{\lambda},\quad \tilde{\lambda} = e^{y^2}
\ela
then we obtain the 0+3-dimensional 2-component system
\bela{Z54}
  p_{y^4} + q_{y^1} = 0,\quad q_{y^3} = p_{y^1}q.
\ela
Elimination of $p$ between these two equations leads to the Boyer-Finley heavenly equation
\bela{Z55}
  \varphi_{y^3y^4} = \left({e^{\varphi}}\right)_{y^1y^1},\quad q = -e^\varphi
\ela
which governs self-dual Einstein spaces subject to the existence of a Killing vector \cite{BoyerFinley1982} and may be regarded as the continuum limit of the integrable 2-dimensional Toda lattice equation. Moreover, the pair  (\ref{Z52}) is readily seen to constitute the associated standard Lax pair (see, e.g., \cite{BogdanovKonopelchenko2013}).

\subsection*{(c) The dispersionless KP equation}

Another admissible choice for the functions $a$ and $b$ in the reduction (\ref{Z50}) is given by
\bela{Z56}
 a = -u_{y_1} + \frac{1}{2}\big(y^2\big)^2,\quad b = 
y^2u_{y^1} + u_{y^3} - \frac{1}{3}\big(y^2\big)^3,
\ela
where $u=u(y^1,y^3,y^4)$. Indeed, evaluation of (\ref{Z51}) produces the dispersionless KP equation
\bela{Z57}
u_{y^4y^1} + u_{y^3y^3} = u_{y^1}u_{y^1y^1}.
\ela
As in the case of the Boyer-Finley equation, the Lax pair (\ref{Z52}) contains derivatives with respect to the spectral parameter $y^2$ (see, e.g., \cite{ManakovSantini2007}).

\subsection*{(d) The dispersionless Hirota equation}

The homogeneity of the general heavenly equation may be exploited to obtain an autonomous 3-dimensional equation by setting
\bela{Z58}
  \Theta = a(x^1)\Xi(x^2,x^3,x^4).
\ela
The symmetric equation (\ref{Z6}) then reduces to the dispersionless Hirota equation
\bela{Z59}
 \begin{split}
  & \phantom{+}\,\,\, (\lambda_1-\lambda_2)(\lambda_3-\lambda_4)\Xi_{x^2}\Xi_{x^3x^4}\\
  & +  (\lambda_2-\lambda_3)(\lambda_1-\lambda_4)\Xi_{x^4}\Xi_{x^2x^3}\\  
  & +  (\lambda_3-\lambda_1)(\lambda_2-\lambda_4)\Xi_{x^3}\Xi_{x^2x^4} = 0
\end{split}
\ela
which has been generalised in \S 3(d) to 3+3 dimensions (cf.\ (\ref{Z33b})) in connection with a symmetry constraint. Accordingly, the multi-dimensional consistency of the dispersionless Hirota equation, which has originally been observed in \cite{Krynski2018}, is encoded in that of its 3+3-dimensional version. Thus, the dispersionless Hirota equation is compatible with another three equations of the same kind which, in the current setting, are given by
\bela{Z60}
 \begin{split}
  & \phantom{+}\,\,\, (\lambda_1-\lambda_k)(\lambda_l-\lambda_5)\Xi_{x^k}\Xi_{x^lx^5}\\
  & +  (\lambda_k-\lambda_l)(\lambda_1-\lambda_5)\Xi_{x^5}\Xi_{x^kx^l}\\  
  & +  (\lambda_l-\lambda_1)(\lambda_k-\lambda_5)\Xi_{x^l}\Xi_{x^kx^5} = 0
\end{split}
\ela
for $k \neq l \in\{2,3,4\}$. These may be obtained directly by applying the reduction (\ref{Z58}) with \mbox{$\Xi=\Xi(x^2,x^3,x^4,x^5)$} to the system (\ref{Z7})$|_{i=1}$ of compatible general heavenly equations. Moreover, in the case $(i,k,l,m)=(2,3,4,5)$, the corresponding general heavenly equation
\bela{Z61}
 \begin{split}
  & \phantom{+}\,\,\, (\lambda_2-\lambda_3)(\lambda_4-\lambda_5)\Xi_{x^2x^3}\Xi_{x^4x^5}\\
  & +  (\lambda_3-\lambda_4)(\lambda_3-\lambda_5)\Xi_{x^3x^4}\Xi_{x^2x^5}\\  
  & +  (\lambda_4-\lambda_2)(\lambda_3-\lambda_5)\Xi_{x^4x^2}\Xi_{x^3x^5} = 0
\end{split}
\ela
persists. Accordingly, any simultaneous solution $\Xi(x^2,x^3,x^4,x^5)$ of the compatible system of four 3-dimensional dispersionless Hirota equations obeys the general heavenly equation which, by construction, constitutes an algebraic consequence of the four dispersionless Hirota equations. It is important to note that this is exactly the algebraic phenomenon which is well known for the discrete Hirota equation since the latter and its compatible counterparts are formally obtained by replacing the derivatives in (\ref{Z59}) and (\ref{Z60}) by shifts on a $\Z^4$-lattice (cf.\ \S 8(b)). 

\subsection*{(e) The Mikhalev equation}

The analogue of the ansatz (\ref{Z58}) applied to the 6-dimensional general heavenly equation (\ref{Z48}) given by
\bela{Z61a}
  \Xi = y^1 u(x^3,x^4,y^2,y^3,y^4)
\ela
leads to the 2+3-dimensional generalisation
\bela{Z61b}
  u_{x^3y^4} - u_{x^4y^3} = u_{x^3}u_{x^4y^2} - u_{x^4}u_{x^3y^2}
\ela
of the Mikhalev equation \cite{Mikhalev1992}. Indeed, application of the symmetry constraints $u_{x^3}=u_{y^2}$ and $u_{x^4}=u_{y^3}$ produces the 3-dimensional equation
\bela{Z62b}
  u_{y^2y^4} - u_{y^3y^3} = u_{y^2}u_{y^2y^3} - u_{y^3}u_{y^2y^2}
\ela
which, apart from the setting of \cite{Mikhalev1992}, also arises is various other contexts such as exceptional hydrodynamic-type systems~\cite{Pavlov2018}, self-dual Einstein spaces and hydrodynamic chains (see \cite{ManakovSantini2007} and references therein). The associated Lax pair $D_3\psi = -u_{y^2}\psi_{y^2}$, $D_4\psi = -u_{y^3}\psi_{y^2}$ with $\psi_{x^3}=\psi_{y^2}$ and $\psi_{x^4}=\psi_{y^3}$, obtained from (\ref{Z49}) by setting $\psi_{y^1}=0$, is equivalent to the known Lax pair for the Mikhalev equation (see, e.g., \cite{ManakovSantini2007}).

\subsection*{(f) Plebanski's first heavenly equation}

Another ansatz of the form (\ref{Z47}), namely 
\bela{Z62}
  \Xi = x^3y^4 + \Lambda(x^3,x^4,y^1,y^2),
\ela
reduces the 6-dimensional second heavenly equation to
\bela{Z63}
 \Lambda_{x^3y^1} \Lambda_{x^4y^2} - \Lambda_{x^3y^2}\Lambda_{x^4y^1} = 1.
\ela
This 2+2-dimensional equation is the first heavenly equation set down by Plebanski \cite{Plebanski1975}. On setting $\psi_{y^3}=\psi_{y^4}=0$, the associated linear system (\ref{Z13}) becomes
\bela{Z64}
  \left(\bear{cccc} 0 & 1 & \Lambda_{x^4y^2} & -\Lambda_{x^3y^2}\\ -1 & 0 & -\Lambda_{x^4y^1} & \Lambda_{x^3y^1}\\
                    - \Lambda_{x^4y^2} & \Lambda_{x^4y^1} & 0 & 1\\ \Lambda_{x^3y^2} &- \Lambda_{x^3y^1} & -1 & 0\ear\right)
  \left(\bear{c}\psi_{y^1}\\ \psi_{y^2}\\ -\lambda\psi_{x^3}\\ -\lambda\psi_{x^4}\ear\right)  = 0
\ela
with the corresponding Lax pair consisting of the third and fourth equation being standard \cite{DoubrovFerapontov2010}.

The infinitesimal symmetry of the 4+4-dimensional \TED equation stated in Theorem \ref{symmetry} now reveals in which {\em asymmetric} sense the first heavenly equation is multi-dimensionally consistent. Indeed, if we set $\psi=\Theta_s$ in (\ref{Z64}) then we obtain four equations which are essentially the same but different from the first Plebanski equation. Indeed, the four equations 
\bela{Z65}
\begin{split}
 \lambda\Lambda_{sx^3} & = \Lambda_{x^3y^1}\Lambda_{zy^2} - \Lambda_{x^3y^2}\Lambda_{zy^1}\\
 \lambda\Lambda_{sx^4} & = \Lambda_{x^4y^1}\Lambda_{zy^2} - \Lambda_{x^4y^2}\Lambda_{zy^1}\\
 \lambda^{-1}\Lambda_{sy^1} & = \Lambda_{x^4y^1}\Lambda_{zx^3} - \Lambda_{x^3y^1}\Lambda_{zx^4}\\
 \lambda^{-1}\Lambda_{sy^2} & = \Lambda_{x^4y^2}\Lambda_{zx^3} - \Lambda_{x^3y^2}\Lambda_{zx^4}
 \end{split}
\ela
are all of Husain-Park type \cite{DoubrovFerapontov2010}. A possible connection between the above system of five compatible heavenly equations and Takasaki's hyper-K\"ahler hierarchy associated with Plebanski's first heavenly equation \cite{Takasaki1989} is currently under investigation. 

\subsection*{(g) The Husain-Park equation}

Both the first heavenly equation and the Husain-Park equation may be regarded as degenerations of the general heavenly equation so that the consistency feature established in the preceding may also be inferred from the fully symmetric multi-dimensional consistency of the general heavenly equation. It turns out enlightening to present a sketch of this alternative derivation. Thus, one may directly verify that the symmetric system (\ref{Z7}) of general heavenly equations degenerates to
\bela{Z66}
 \begin{split}
  (\mu_3-\mu_4)\Xi_{x^3x^4} &= \Xi_{x^2x^3}\Xi_{x^1x^4} -  \Xi_{x^2x^4}\Xi_{x^1x^3}\\
  (\mu_4-\mu_5)\Xi_{x^4x^5} &= \Xi_{x^2x^4}\Xi_{x^1x^5} -  \Xi_{x^2x^5}\Xi_{x^1x^4}\\
  (\mu_5-\mu_3)\Xi_{x^5x^3} &= \Xi_{x^2x^5}\Xi_{x^1x^3} -  \Xi_{x^2x^3}\Xi_{x^1x^5}
 \end{split}
\ela
together with 
\bela{Z67}
 \begin{split}
   (\mu_3-\mu_4)\Xi_{x^3x^4}\Xi_{x^1x^5} +  (\mu_4 -\mu_5)\Xi_{x^4x^5}\Xi_{x^1x^3} +  (\mu_5-\mu_3)\Xi_{x^5x^3}\Xi_{x^1x^4} &= 0\\
   (\mu_3-\mu_4)\Xi_{x^3x^4}\Xi_{x^2x^5} +  (\mu_4 -\mu_5)\Xi_{x^4x^5}\Xi_{x^2x^3} +  (\mu_5-\mu_3)\Xi_{x^5x^3}\Xi_{x^2x^4} &= 0
 \end{split}
\ela
in the limit
\bela{Z68}
 \bear{c}
 \lambda_3 = \mu_3 + \sigma,\quad \lambda_4 = \mu_4 + \sigma,\quad \lambda_5 = \mu_5 + \sigma\as
 \Theta = \sigma x^1 x^2 + \Xi,\quad \sigma\rightarrow\infty.
 \ear
\ela
Here, we have applied without loss of generality the normalisation $\lambda_1\rightarrow\infty$ and $\lambda_2=0$, corresponding to a ``base equation'' of the form (\ref{Z1}). Accordingly, by construction, the symmetric ``hyperbolic'' triple (with respect to $x^3$, $x^4$ and $x^5$) of compatible Husain-Park equations (\ref{Z66}) is compatible and its solutions also satisfy the symmetric pair (with respect to $x^1$ and $x^2$) of general heavenly equations (\ref{Z67}). 

The above phenomenon is reminiscent of the multi-dimensional consistency of the Schwarzian version of the Hirota (dSKP) equation which constitutes a multi-ratio condition of Menelaus type \cite{KonopelchenkoSchief2002} in that the dSKP equation is multi-dimensionally consistent ``with itself'' but its degeneration to the discrete modified KP equation is not. Its asymmetric multi-dimensional consistency necessarily involves the dSKP equation \cite{AdlerBobenkoSuris2012}. It is now readily confirmed that a deeper degeneration is obtained by considering the additional limit
\bela{Z69}
 \Xi = \rho x^3x^4 + \Lambda,\quad  \mu_3-\mu_4 = \frac{1}{\rho}(\kappa_3-\kappa_4),\quad \rho\rightarrow\infty\as
\ela
leading to the compatibility of the first Plebanski and four Husain-Park equations discussed in the preceding section. It is tempting to perform similar limits on the remaining two Husain-Park equations (\ref{Z66})$_{2,3}$. This is indeed possible and (\ref{Z66}) then reduces to a compatible triple of first heavenly equations. However, in these limits, the two general heavenly equations (\ref{Z67}) become linear which implies that the three first heavenly equations are, in fact, equivalent to the initial first Plebanski equation and these two linear equations so that this (asymmetric) multi-dimensional extension of the first Plebanski equation becomes trivial.

The results of a more detailed investigation of the multi-dimensional consistency property of reductions of the 4+4-dimensional \TED equation will be presented elsewhere. Here, we merely observe that the multi-dimensional consistency of, for instance, the second heavenly equation involves both 4- and 5-dimensional equations.

\section{\boldmath A canonical extension to $2n$+$2n$ dimensions}

Remarkably, it turns out that there exists a canonical way of generalising the 4+4-dimensional \TED equation to $2n$+$2n$ dimensions in such a manner that all of the properties established in \S 3 and \S 4 are preserved. This extension is based on a ``miraculous'' classical algebraic property of the ``Pfaffian adjugate'' which guarantees that the skew-symmetric matrix $A$ underlying the generalisation of the symmetric linear system (\ref{Z13}) to arbitrary even order has rank 2 when imposing a single Pfaffian condition. As a result, one may interpret the $2n$+$2n$-dimensional \TED equation as being generated by the commutativity of two vector fields. Moreover, the proof of the theorem stated below does not require any new ideas. It follows exactly the line of arguments presented in \S 3 and \S 4.

\subsection*{\boldmath (a) A $2n$+$2n$-dimensional \TED equation}

In the preceding, it has been demonstrated that the 4+4-dimensional \TED equation may be formulated as the vanishing of the Pfaffian of the 2$\times$2 skew-symmetric matrix $\omega$, that is,
\bela{Z70}
  \pf(\omega) = \omega_{12}\omega_{34} + \omega_{23}\omega_{14} + \omega_{31}\omega_{24} = 0.
\ela
In general, the Pfaffian of a $2n$$\times$$2n$ skew-symmetric matrix $\omega$ is (implicitly) given by \cite{Hirota2003}
\bela{Z71}
  \frac{1}{n!}\wedge^n\bar{\Omega} = \pf(\omega) e^1\wedge\cdots\wedge e^{2n},
\ela
where the 2-form $\bar{\Omega}$ is defined by
\bela{Z72}
 \bar{\Omega} =\frac{1}{2} \omega_{ik}e^i\wedge e^k
\ela
and $\{e^1,\ldots,e^{2n}\}$ denotes the standard basis of $\mathbb{R}^{2n}$. Hence, for instance, the 4+4-dimensional \TED equation may be formulated as $\wedge^2\bar{\Omega}  = \bar{\Omega}\wedge\bar{\Omega} = 0$. In terms of the Levi-Civita symbol, the Pfaffian adopts the form
\bela{Z73}
  \pf(\omega) = \frac{1}{2^n n!}\epsilon^{i_1\cdots i_{2n}}\omega_{i_1i_2}\cdots\omega_{i_{2n-1}i_{2n}},
\ela
which coincides with (\ref{Z70}) in the case $n=2$.

Now, the key observation is that the skew-symmetric matrices $A$ and $\omega$ in terms of which the linear system (\ref{Z16}) is defined obey the relation
\bela{Z74}
  A^{kl}\omega_{lm} = -2\pf(\omega){\delta^k}_m,
\ela
where ${\delta^k}_m$ denotes the Kronecker symbol. Hence, $A$ may be regarded as the scaled Pfaffian analogue of the adjugate of $\omega$. In fact, this is just a particular case of the general classical identity~\cite{Shapiro2000} 
\bela{Z75}
  A^{kl}\omega_{lm} = -2^{n-1}(n-1)!\pf(\omega){\delta^k}_m
\ela
which holds for any $2n$$\times$$2n$ skew-symmetric matrix $\omega$ and the scaled Pfaffian adjugate $A$ given by
\bela{Z76}
  A^{kl} =   \epsilon^{i_1\cdots i_{2n-2}kl}\omega_{i_1i_2}\cdots\omega_{i_{2n-3}i_{2n-2}}.
\ela
In matrix notation, this may be formulated as $A\omega \sim\pf(\omega)\one$.
One of the remarkable implications of this identity is that if the matrix $\omega$ is singular but $\pf(\omega)=0$ is the only constraint on $\omega$ then $\rank\omega= 2n-2$ together with $A\omega=0$ yields \mbox{$\rank A\leq 2$}. In fact, since $\det A=0\,\Leftrightarrow\,\det\omega=0$ and the coefficients $A^{kl}$ are essentially Pfaffians which are signed square roots of $(2n-2)$$\times$$(2n-2)$ principal minors of the matrix $\omega$, it is evident that the important equivalence
\bela{Z77}
  \rank A = 2\quad\Leftrightarrow\quad \rank\omega=2n-2
\ela
holds.
Accordingly, if we consider the ``generic'' situation within the class $\pf(\omega)=0$, that is, $\rank\omega=2n-2$, then the following theorem holds.

\begin{theorem}
The linear system of $2n$ equations
\bela{Z78}
  A^{kl}D_l\psi = 0,
\ela
wherein the skew-symmetric matrix $A$ is given by
\bela{Z79}
  A^{kl} =   \epsilon^{i_1\cdots i_{2n-2}kl}\omega_{i_1i_2}\cdots\omega_{i_{2n-3}i_{2n-2}}
\ela
and
\bela{Z80}
  \omega_{ik} = \Theta_{x^iy^k} - \Theta_{x^ky^i},\quad D_i= \del_{y^i} - \lambda\del_{x^i},
\ela
is compatible modulo the $2n$+$2n$-dimensional \TED equation $\pf(\omega)=0$, that is,
\bela{Z81}
  \epsilon^{i_1\cdots i_{2n}}\omega_{i_1i_2}\cdots\omega_{i_{2n-1}i_{2n}} = 0.
\ela
The latter guarantees that $\rank A = 2$ so that the $2n$+$2n$-dimensional \TED equation is generated by the commutativity $[X^k,X^{k'}]=0\bmod\{X^k,X^{k'}\}$ of any pair of the $2n$ divergence-free vector fields $X^k = A^{kl}D_l$. The eigenfunction $\psi$ constitutes an infinitesimal symmetry and the $2n$+$2n$-dimensional \TED equation is multi-dimensionally consistent in that the \mbox{$(2n+1)$$+$$(2n+1)$}-dimensional system of $2n$$+$$2n$-dimensional \TED equations
\bela{Z82}
 H^k = \epsilon^{i_1\cdots i_{2n}k}\omega_{i_1i_2}\cdots\omega_{i_{2n-1}i_{2n}} = 0,\quad k=1,\ldots,2n+1
\ela
is in involution.
\end{theorem}

\begin{remark}
As in the 4+4-dimensional case, by virtue of the definition of the Pfaffian, the $2n$+$2n$-dimensional \TED equation may be expressed as 
\bela{ZZ83a}
  \del\Omega = 0,\quad \wedge^n\Omega=0, \qquad  \Omega = \omega_{ik}dz^i\wedge dz^k,
\ela
where the coefficients of the differential 2-form are, once again, real. We also note that any solution of the \TED equation for which the complex Hessian ${(\Theta_{z^i\bar{z}^k})}_{i,k}$ is positive definite gives rise to a K\"ahler metric \cite{BlumenhagenLuestTheisen2013} via the real closed 2-form
\bela{ZZ84d}
  \tilde{\omega} = ig_{kl}dz^k\wedge d\bar{z}^l,\quad g_{kl} = \Theta_{z^k\bar{z}^l}.
\ela
\end{remark}

\begin{remark}
The above formulation in terms of a differential 2-form is reminiscent of the construction of the $2n$-dimensional generalisation of the general heavenly equation proposed by Bogdanov \cite{Bogdanov2015} which corresponds to the travelling wave reduction $\Theta_{y^i}=\lambda_i\Theta_{x^i}$. However, the essential difference in the current approach is that the matrix $\omega$ is {\em independent} of the spectral parameter $\lambda$. Here, the dependence on $\lambda$ is encapsulated in the operators $D_i$ which is the key to ``doubling'' the number of independent variables and revealing the nature of the parameters~ $\lambda_i$.
\end{remark}

\begin{remark}
Since the vector fields $X^k=A^{kl}D_l$ are formally anti-self adjoint due to being divergence-free, we immediately conclude that
\bela{ZZ85}
  D_l(A^{kl}\psi) = 0.
\ela
The significance of these $\lambda$-dependent conservations laws in terms of the \TED equation will be discussed elsewhere.
\end{remark}

\begin{remark}\label{lower}
Even though we have assumed that $\rank\omega=2n-2$ so that $A$ does not vanish identically, the solutions of the $2n$+$2n$-dimensional \TED equation for which $A=0$ are still of significance. Indeed, the multi-dimensional consistency of the $2k$+$2k$-dimensional \TED equation for any $k<n$ implies that any solution of its consistent extension to a system of $2k$+$2k$-dimensional equations involving $2n$+$2n$ independent variables obeys the $2n$+$2n$-dimensional \TED equation. In this connection, it is observed that even the case $k=1$ is meaningful since the associated 2+2-dimensional \TED equation becomes the 2+2-dimensional wave equation
\bela{Z83}
 \Theta_{x^1y^2} - \Theta_{x^2y^1} = 0,
\ela
the multi-dimensional consistency of which is trivial due to its linearity.
\end{remark}

\subsection*{(b) A Manakov-Santini connection}

Both the dispersionless KP equation and the Mikhalev equation are known to represent prominent reductions of a multi-dimensional system established by Manakov and Santini \cite{ManakovSantini2006} in the spirit of Zakharov and Shabat's original work \cite{ZakharovShabat1979} on multi-dimensional integrable systems generated by the commutativity requirement of two vector fields. Since Plebanski's second heavenly equation has also been shown to be a Hamiltonian reduction of the Manakov-Santini system \cite{ManakovSantini2006}, it is natural to investigate the existence of a possible link between the Manakov-Santini system and the \TED equation. Here, we briefly indicate that, at least in the 2+4-dimensional case, a connection may be made.

The general scalar Lax pair based on two normalised formal vector fields is given by
\bela{M1}
 (D_{n-1} + u^iD_i)\psi = 0,\quad  (D_n + v^iD_i)\psi = 0
\ela
with Einstein's summation convention covering indices from $1$ to $n-2$. The associated compatibility condition leads to the underdetermined system of equations
\bela{M2}
  D_nu^i - D_{n-1}v^i + v^kD_ku^i - u^kD_kv^i = 0.
\ela
If we make the usual choice
\bela{M3}
  D_i = \del_{y^i} - \lambda\del_{x^i},\quad i=1,\ldots, n
\ela
then the above system splits into the two sets
\bela{M4}
 \begin{split}
 u^i_{x^n} - v^i_{x^{n-1}} + v^ku^i_{x^k} - u^kv^i_{x^k} & = 0\\
 u^i_{y^n} - v^i_{y^{n-1}} + v^ku^i_{y^k} - u^kv^i_{y^k} & = 0.
 \end{split}
\ela
The Manakov-Santini system is obtained by assuming independence of the variables $x^i$,\linebreak \mbox{$i=1,\ldots,n-2$} so that
\bela{M5}
 u^i_{x^n} - v^i_{x^{n-1}} = 0,\quad
 u^i_{y^n} - v^i_{y^{n-1}} + v^ku^i_{y^k} - u^kv^i_{y^k} = 0,
\ela
corresponding to the Lax pair
\bela{M6}
 \begin{split}
  (\del_{y^{n-1}} - \lambda\del_{x^{n-1}} + u^i\del_{y^i})\psi &= 0\\
  (\del_{y^{n}} - \lambda\del_{x^{n}} + v^i\del_{y^i})\psi &= 0.
 \end{split}
\ela
Moreover, it has been demonstrated that it is consistent to impose the condition that the vector fields underlying the above Lax pair be divergence-free so that
\bela{M7}
  \del_{y^i}u^i = 0,\quad \del_{y^i}v^i = 0.
\ela

The generality of the scheme (\ref{M1}), (\ref{M3}) raises the question as to whether there exists a connection between the vanishing divergence conditions leading to the \TED equation and the Manakov-Santini system subject to the admissible divergence conditions (\ref{M7}). It is important to note that even though the Lax pairs associated with vector fields are invariant under multiplication of the vector fields by arbitrary functions, the condition that the vector fields be divergence-free is not. Accordingly, it is {\em a priori} not evident that such a connection exists and this problem will be investigated in detail elsewhere. Here, we focus on the case \mbox{$n = 4$} so that the first set of equations (\ref{M5})$_1$ together with the vanishing divergence conditions adopt the form
\bela{M8}
  u^1_{x^4} = v^1_{x^3},\quad  u^2_{x^4} = v^2_{x^3},\quad u^1_{y^1} + u^2_{y^2} = 0,\quad v^1_{y^1} + v^2_{y^2} = 0
\ela
which may be resolved in terms of a single function $\Xi$ according to
\bela{M9}
  u^1 = -\Xi_{x^3y^2},\quad u^2 = \Xi_{x^3y^1},\quad v^1 = -\Xi_{x^4y^2},\quad v^2 = \Xi_{x^4y^1}.
\ela
Now, comparison with (\ref{Z49}) reveals that the Lax pair (\ref{M6}) for $n=4$ is identical with the Lax pair for the 6-dimensional second heavenly equation (\ref{Z48}). Indeed, the remaining equations (\ref{M5})$_2$ turn out to be differential consequences of the 6-dimensional second heavenly equation and are, in fact, equivalent to it modulo the additive gauge freedom encoded in the parametrisation (\ref{M9}).

\section{\boldmath Embedding into $(2n+1)$+$(2n+1)$ dimensions}

It has been pointed out in Remark \ref{lower} that the $2n$+$2n$-dimensional \TED equation admits solutions in terms of compatible systems of lower-dimensional \TED equations. Alternatively, the ansatz
\bela{Z84}
 \Theta = x^{2n+1}y^{2n+2} + \tilde{\Theta}(x^1,\ldots,x^{2n},y^1,\ldots,y^{2n})
\ela
reduces the $(2n+2)$+$(2n+2)$-dimensional \TED equation to the $2n$+$2n$-dimensional \TED equation. Since the above ansatz does not involve the variables $x^{2n+2}$ and $y^{2n+1}$, it is natural to refer to the dimensional reduction
\bela{Z85}
  \epsilon^{i_1\cdots i_{2n+2}}\omega_{i_1i_2}\cdots\omega_{i_{2n+1}i_{2n+2}} = 0, \quad \Theta_{x^{2n+2}}=0,\quad\Theta_{y^{2n+1}}=0
\ela
as the $(2n+1)$+$(2n+1)$-dimensional \TED equation and regard the $2n$+$2n$-dimensional \TED equation as being embedded in the $(2n+1)$+$(2n+1)$-dimensional \TED equation via the relation (\ref{Z84}). In particular, any of the reductions of the 4+4-dimensional \TED equation discussed in \S 5 may be ``lifted'' to a corresponding reduction of the 5+5-dimensional \TED equation which we may formulate as
\bela{Z86}
  \epsilon^{iklmpq}\Theta_{x^iy^k}\Theta_{x^ly^m}\Theta_{x^py^q}=0,\quad \Theta_{x^6}=0,\quad \Theta_{y^5}=0.
\ela
In this connection, it turns out convenient to introduce the ``two-directional Hessian'' (determinant)
\bela{Z87}
 \mathcal{H}(\Theta;x^{i_1},\ldots,x^{i_m};y^{k_1},\ldots,y^{k_m}) = 
   \left|\bear{ccc} \Theta_{x^{i_1}y^{k_1}} & \cdots & \Theta_{x^{i_1}y^{k_m}}\\
                            \vdots & \ddots & \vdots\\
                            \Theta_{x^{i_m}y^{k_1}} & \cdots & \Theta_{x^{i_m}y^{k_m}}
   \ear\right|
\ela
as will be seen below.

\subsection*{(a) A 3+5-dimensional second heavenly equation}

We now apply the reduction (\ref{Z47}) to the 5+5-dimensional \TED equation (\ref{Z86}) so that
\bela{Z88}
  \Theta = x^1y^2 + \Xi(x^3,x^4,x^5,y^1,y^2,y^3,y^4,y^6),
\ela
bearing in mind the additional independent variables $x^5$ and $y^6$. As a result, we obtain the 3+5-dimensional extension
\bela{Z89}
  \mathcal{H}(\Xi;x^3,x^5;y^4,y^6) -  \mathcal{H}(\Xi;x^4,x^5;y^3,y^6) =  \mathcal{H}(\Xi;x^3,x^4,x^5;y^1,y^2,y^6)
\ela
of the 6-dimensional second heavenly equation. Indeed, by construction, the specialisation
\bela{Z90}
  \Theta = x^5y^6 + \tilde{\Xi}(x^3,x^4,y^1,y^2,y^3,y^4)
\ela
leads to (\ref{Z48}).

\subsection*{(b) A 5-dimensional second heavenly equation}

The identification $(x^3,x^4,x^5)=(y^1,y^2,y^6)$, corresponding to a symmetry reduction of the 3+5-dimensional second heavenly equation (\ref{Z89}), generates the 5-dimensional equation
\bela{Z91}
  \mathcal{H}(\Xi;x^3,x^5;y^4,x^5) -  \mathcal{H}(\Xi;x^4,x^5;y^3,x^5) =  \mathcal{H}(\Xi;x^3,x^4,x^5;x^3,x^4,x^5).
\ela
It is observed that the last term in this equation constitutes a proper cubic Hessian determinant. The specialisation
\bela{Z92}
  \Xi = {(x^5)}^2 + \tilde{\Xi}(x^3,x^4,y^3,y^4)
\ela
gives rise to Plebanski's second heavenly equation
\bela{Z93}
  \tilde{\Xi}_{x^3y^4} - \tilde{\Xi}_{x^4y^3} =  \mathcal{H}(\tilde{\Xi};x^3,x^4,x^3,x^4)
\ela
in accordance with the general approach (\ref{Z84}).

\subsection*{(c) A 3+3-dimensional first heavenly equation}

The analogue of the deeper reduction (\ref{Z62}), that is,
\bela{Z94}
  \Xi = x^3y^4 + \Lambda(x^3,x^4,x^5,y^1,y^2,y^6),
\ela
applied to the 3+5-dimensional second heavenly equation (\ref{Z89}), produces the cubic 3+3-dimen\-sional version 
 \bela{Z95}
 \Lambda_{x^5y^6}  = \mathcal{H}(\Lambda;x^3,x^4,x^5;y^1,y^2,y^6)
\ela
of Plebanski's first heavenly equation. It is noted that the latter may be formulated as
\bela{Z96}
  1 = \mathcal{H}(\tilde{\Lambda};x^3,x^4;y^1,y^2)
\ela
which confirms the connection $\Lambda = x^5y^6 + \tilde{\Lambda}(x^3,x^4,y^1,y^2)$.

\subsection*{(d) A 2+4-dimensional first heavenly equation}

A quadratic generalisation of Plebanski's first heavenly equation is obtained by considering the reduction
\bela{Z97}
  \Theta = x^1y^2 + x^3y^6 + x^5y^4 + \Lambda(x^4,x^5,y^1,y^2,y^3,y^6)
\ela
applied to the 5+5-dimensional \TED equation subject to $\Theta_{x^2}=0$. In this case, we deduce that
\bela{Z98}
 \mathcal{H}(\Lambda;x^4,x^5;y^1,y^2) + \mathcal{H}(\Lambda;x^4,x^5;y^3,y^6) + 1 = 0.
\ela
Even though the usual specialisation $\Lambda = x^5y^6 + \tilde{\Lambda}(x^4,y^1,y^2,y^3)$ becomes trivial, the dimensional reduction $\Lambda_{y^6}=0$ produces the standard first heavenly equation
\bela{Z99}
 \mathcal{H}(\Lambda;x^4,x^5;y^1,y^2) + 1 = 0.
\ela
A systematic examination of the reductions of the higher-dimensional \TED equation is under way. In this connection, it is noted that the 2+4-dimensional first heavenly equation (\ref{Z98}) turns out to be part of Takasaki's hyper-K\"ahler hierarchy associated with Plebanski's first heavenly equation \cite{Takasaki1989}. 

\section{Concluding remarks}

\subsection*{(a) A Grassmannian connection}

Since the rank of the skew-symmetric matrix $A$ is 2, the solutions of the $2n$+$2n$-dimensional \TED equation parametrise, via the $(2n-2)$-dimensional kernel of $A$, ``trajectories'' in the Grassmannian $\Gr(2n-2,2n)$, which is the space of all linear subspaces of dimension $2n-2$ of a $2n$-dimensional vector space \cite{HodgePedoe1994} (cf.\ \cite{BogdanovKonopelchenko2013}). Specifically, in terms of the Pfaffians 
\bela{Z100}
  \pi_{i_1\cdots i_{2n-2}} = \pf(i_1,\ldots,i_{2n-2})
\ela
of the $(2n-2)$$\times$$(2n-2)$ skew-symmetric submatrices of $\omega$ with the indicated ordering of the $2n-2$ distinct indices, the linear system (\ref{Z78}), (\ref{Z79}) may be formulated as
\bela{Z101}
  \sum_{k=1}^{2n-1} (-1)^k \pi_{i_1\cdots i_{k-1} i_{k+1} \cdots i_{2n-1}} D_{i_k}\psi = 0
\ela
for all sets of distinct indices $\{i_1,\ldots,i_{2n-1}\}$. Now, by virtue of the Pfaffian identity \cite{Hirota2003}
\bela{Z102}
  \begin{split}
  &\phantom{+}\,\,\pf(1,2,3,4,5,\ldots,2n)\pf(5,\ldots,2n)\\
& = \pf(1,2,5,\ldots,2n)\pf(3,4,5,\ldots,2n)\\
& + \pf(2,3,5,\ldots,2n)\pf(1,4,5,\ldots,2n)\\
& + \pf(3,1,5,\ldots,2n)\pf(2,4,5,\ldots,2n),
\end{split}
\ela
we conclude that $\pf(\omega)=0$ implies that
\bela{Z103}
  \pi_{12\alpha}\pi_{34\alpha} + \pi_{23\alpha}\pi_{14\alpha} + \pi_{31\alpha}\pi_{24\alpha} = 0,
\ela
where $\alpha$ denotes the multi-index $(5\cdots 2n)$. The complete set of relations of this type obtained by permuting the indices constitute the Pl\"ucker relations  \cite{HodgePedoe1994} associated with the Grassmannian $\Gr(2n-2,2n)$. Accordingly, only $2n - (2n-2) = 2$ equations of the linear system (\ref{Z101}) are independent which confirms the assertion that $\rank A = 2$. Thus, the quantities $\pi_{i_1\cdots i_{2n-2}}$ are the Pl\"ucker coordinates of the Grassmannian $\Gr(2n-2,2n)$. The geometric implications of the connection between the \TED equation and Grassmannians is the subject of a separate investigation.

\subsection*{(b) The discrete general heavenly equation}

It is natural to inquire as to a possible extension of the approach presented here to the discrete setting since the general heavenly equation (\ref{Z6}) has originally been derived \cite{Schief1999,Schief1996} as the continuum limit of the discrete equation
\bela{Z104}
 \begin{split}
  & \phantom{+}\,\,\, (\lambda_1-\lambda_2)(\lambda_3-\lambda_4)(\tau\tau_{12}-\tau_1\tau_2)(\tau\tau_{34}-\tau_3\tau_4)\\
  & +  (\lambda_2-\lambda_3)(\lambda_1-\lambda_4)(\tau\tau_{23}-\tau_2\tau_3)(\tau\tau_{14}-\tau_1\tau_4)\\  
  & +  (\lambda_3-\lambda_1)(\lambda_2-\lambda_4)(\tau\tau_{13}-\tau_1\tau_3)(\tau\tau_{24}-\tau_2\tau_4) = 0
\end{split}
\ela
with the identification $\Theta = \ln\tau$ and the indices on $\tau$ denoting shifts in the respective directions of the underlying $\Z^4$-lattice. This discrete general heavenly equation admits a great variety of integrable reductions as initially indicated in \cite{Schief1999} and analysed in detail in a separate publication. Therein, possible extensions in both the continuous and discrete settings of the algebraic scheme presented here are discussed. 

Even though the discrete general heavenly equation does not, {\em a priori}, appear to be multi-dimensionally consistent, it may be decomposed into the standard system of compatible 3-dimen\-sional Hirota equations, thereby revealing a novel algebraic consequence of the multi-dimensional consistency of the discrete master Hirota equation. The key idea is to split (\ref{Z104}) into two parts which are independent of the ``unshifted'' $\tau$, that is,
\bela{Z105}
  \begin{split}
  & \phantom{+}\,\,\, (\lambda_1-\lambda_2)(\lambda_3-\lambda_4)\tau_{12}\tau_{34}\\
  & +  (\lambda_2-\lambda_3)(\lambda_1-\lambda_4)\tau_{23}\tau_{14}\\  
  & +  (\lambda_3-\lambda_1)(\lambda_2-\lambda_4)\tau_{13}\tau_{24} = 0
\end{split}
\ela
and
\bela{Z106}
  \begin{split}
  & \phantom{+}\,\,\, (\lambda_1-\lambda_2)(\lambda_3-\lambda_4)(\tau_3\tau_4\tau_{12}+\tau_1\tau_2\tau_{34})\\
  & +  (\lambda_2-\lambda_3)(\lambda_1-\lambda_4)(\tau_1\tau_4\tau_{23}+\tau_2\tau_3\tau_{14})\\  
  & +  (\lambda_3-\lambda_1)(\lambda_2-\lambda_4)(\tau_2\tau_4\tau_{13}+\tau_1\tau_3\tau_{24}) = 0.
\end{split}
\ela
The first equation (\ref{Z105}) is a well-known algebraic consequence of the compatible system of Hirota equations \cite{AdlerBobenkoSuris2012}
\bela{Z107}
 \begin{split}
  (\lambda_1-\lambda_2)\tau_3\tau_{12} + (\lambda_2-\lambda_3)\tau_1\tau_{23} + (\lambda_3-\lambda_1)\tau_2\tau_{13} & = 0\\
  (\lambda_1-\lambda_2)\tau_4\tau_{12} + (\lambda_2-\lambda_4)\tau_1\tau_{24} + (\lambda_4-\lambda_1)\tau_2\tau_{14} & = 0\\
  (\lambda_1-\lambda_3)\tau_4\tau_{13} + (\lambda_3-\lambda_4)\tau_1\tau_{34} + (\lambda_4-\lambda_1)\tau_3\tau_{14} & = 0\\
  (\lambda_2-\lambda_3)\tau_4\tau_{23} + (\lambda_3-\lambda_4)\tau_2\tau_{34} + (\lambda_4-\lambda_2)\tau_3\tau_{24} & = 0
 \end{split}
\ela
of which only three are algebraically independent. In fact, it merely constitutes another symmetric version of the Hirota equation. Remarkably, the second equation (\ref{Z106}) is then identically satisfied. Hence, the discrete general heavenly equation (\ref{Z104}) may be decomposed into the five compatible Hirota equations (\ref{Z105}), (\ref{Z107}). The system (\ref{Z107}) is  the discrete analogue of the system of four dispersionless Hirota equations (\ref{Z59}), (\ref{Z60}).
\bigskip

\noindent
{\bf Funding.} This research was supported by the Australian Research Council (DP1401000851).
\medskip

\noindent
{\bf Acknowledgement.} BGK gratefully acknowledges the kind hospitality of the School of Mathematics and Statistics, UNSW, Sydney.

\end{document}